\documentclass{article}

\usepackage{graphicx} 
\usepackage{latexsym}
\usepackage{comment}
\usepackage{amsmath}
\usepackage{amssymb}
\usepackage{algorithm}
\usepackage{algpseudocode}
\usepackage{here}

\usepackage{authblk}
\usepackage{geometry}
\newcommand{\ie}{{\em i.e., }}
\newcommand{\eg}{{\em e.g., }}
\newcommand{\etal}{{\em et.~al.}}

\newcommand{\calL}{\mathcal{L}}
\newcommand{\calC}{\mathcal{C}}
\newcommand{\calS}{\mathcal{S}}
\newcommand{\calG}{\mathcal{G}}

\newcommand{\prl}{P_{\mathit{RL}}}
\newcommand{\pdummy}{P'_{\mathit{RL}}}

\newcommand{\leader}{\mathtt{leader}}
\newcommand{\dist}{\mathtt{distL}}
\newcommand{\bull}{\mathtt{bullet}}
\newcommand{\shield}{\mathtt{shield}}
\newcommand{\signal}{\mathtt{signal}}

\newcommand{\vblank}{\vspace{0.3cm}}

\newcommand{\distl}{d_{L}}
\newcommand{\distr}{d_{R}}

\newcommand{\ect}{\mathit{ECT}}

\newcommand{\seqr}{\mathit{seq}_R}
\newcommand{\seql}{\mathit{seq}_L}
\newcommand{\tdel}{T_{\mathrm{del}}}
\newcommand{\peaceful}{\mathit{Peaceful}}
\newcommand{\modest}{\mathit{Modest}}
\newcommand{\secure}{\mathit{Secure}}

\newcommand{\ex}{\mathbf{E}}
\newcommand{\poly}{\mathit{poly}}
\newcommand{\predef}{\stackrel{\mathrm{def}}{\equiv}}

\newcommand{\relmiddle}[1]{\mathrel{}\middle#1\mathrel{}} 
\newcommand{\longmid}{\relmiddle{|}}

\newcommand{\sch}{\mathbf{\Gamma}}

\newcommand{\outputs}{\pi_{\mathit{out}}} 

\newcommand{\call}{\mathcal{C}_{\mathrm{all}}}

\newcommand{\cpb}{\mathcal{C}_{\mathrm{PB}}}

\newcommand{\cni}{\mathcal{C}_{\mathrm{NI}}}
\newcommand{\srl}{\mathcal{S}_{\mathrm{RL}}}

\usepackage{amsthm}
\newtheorem{theorem}{Theorem}
\newtheorem{lemma}{Lemma}

\newtheorem{clm}{Claim}
\theoremstyle{definition}
\newtheorem{dfn}{Definition}

\algtext*{EndWhile}

\begin{document}
\title{Time-Optimal Self-Stabilizing Leader Election on Rings in Population Protocols
\thanks{
This work was supported by JSPS KAKENHI Grant Numbers 19H04085 and 20H04140.
}
}
%
%

\author[1]{Daisuke Yokota}
\author[1]{Yuichi Sudo\thanks{Corresponding author:y-sudou[at]ist.osaka-u.ac.jp}}
\author[1]{Toshimitsu Masuzawa}

\affil[1]{Osaka University, Japan}

\date{}

\maketitle              
\begin{abstract}
%
We propose a self-stabilizing leader election protocol on directed rings in the model of population protocols.
Given an upper bound $N$ on the population size $n$, 
the proposed protocol elects a unique leader
within $O(nN)$ expected steps starting from any configuration
and uses $O(N)$ states. 
This convergence time is optimal if a given upper bound $N$
is asymptotically tight, \ie $N=O(n)$.

\end{abstract}

\section{Introduction}
We consider the \emph{population protocol} (PP) model \cite{AAD+06} in
 this paper.
A network called \emph{population} consists of a large number of
 finite-state automata,
called \emph{agents}.
Agents make \emph{interactions}
 (\ie pairwise communication) with each other
to update their states.
The interactions are opportunistic,
\ie they are unpredictable for the agents.
A population is modeled by a graph $G=(V,E)$,
where $V$ represents the set of agents,
and $E$ indicates which pair of agents can interact.
Each pair of agents $(u,v)\in E$ has interactions
infinitely often,
while each pair of agents $(u',v') \notin E$ never has an interaction.
At each time step,
one pair of agents chosen uniformly at random from all pairs in $E$ has an interaction.
This assumption enables us to evaluate time complexities
of population protocols.\footnote{
Almost all studies in the population protocol model
make this assumption when they evaluate time complexities
of population protocols.
}
In the field of population protocols, many efforts have been devoted to devising protocols for a complete graph,
\ie a population where every pair of agents interacts infinitely often.
In addition, several studies
\cite{AAF+08,AAD+06,BBB13,CP07,CC19,CG17,MNRS14,SOK+14,SMD+16,SOK+20det,SOK+18}
have investigated 
populations 
forming graphs other than complete graphs.

\emph{Self-stabilization} \cite{Dij74}
is a fault-tolerant property whereby,
even when any transient fault
(\eg memory crash) occurs,
the network can autonomously recover from the fault.
Formally, self-stabilization is defined as follows:
(i) starting from an arbitrary configuration,
a network eventually reaches a \emph{safe configuration} (\emph{convergence}),
and
(ii) once a network reaches a safe configuration,
it maintains its specification forever (\emph{closure}).
Self-stabilization is of great importance in the PP model
because self-stabilization tolerates any finite number of transient faults, and this is a necessary property
in a network consisting of a large number of inexpensive and unreliable nodes. 

Consequently,
many studies have been devoted to self-stabilizing population protocols \cite{AAF+08,BBB13,SIW12,CP07,CC19,FJ06,Izu15,SOK+14,SEI+20,SMD+16,SNY+12,SOK+20det,SOK+18,SOK+20polylog}.
For example, Angluin \etal~\cite{AAF+08} proposed self-stabilizing protocols for a variety of problems, \ie leader election in rings, token circulation in rings with a pre-selected leader, 2-hop coloring in degree-bounded graphs, consistent global orientation in undirected rings, and spanning-tree construction in regular graphs.
Sudo \etal~\cite{SMD+16,SOK+18} gave self-stabilizing 2-hop coloring protocol that uses a much smaller memory space of agents. 
Sudo \etal~\cite{SSN+20} investigates the solvability of self-stabilizing leader election, self-stabilizing ranking, self-stabilizing degree recognition, and self-stabilizing neighbor recognition
on arbitrary graphs.


Many of the above studies on self-stabilizing population protocols have focused on self-stabilizing leader election (SS-LE) because
leader election is one of the most fundamental problems in the PP model.\footnote{
Several important protocols \cite{AAF+08,AAD+06,AAE08} require a pre-selected unique leader.
In particular, Angluin \etal~\cite{AAE08}
show that 
all semi-linear predicates can be solved very quickly if we have a unique leader.
}
The goal of the leader election problem
is electing exactly one agent as a leader in the population.
Unfortunately, SS-LE is impossible to solve
without an additional assumption
even if we focus only on
complete graphs \cite{AAF+08,SIW12,SSN+20}. 
The studies to overcome this impossibility in the literature
are roughly classified into four categories.
(Some of the studies belong to multiple categories.)

The first category \cite{Bur+19,SIW12,SSN+20} assumes that every agent knows the exact number of agents.
With this assumption,
Cai \etal~\cite{SIW12} gave an SS-LE protocol for complete graphs, Burman \etal~\cite{Bur+19} gave faster protocols in the same setting,
and Sudo \etal~\cite{SSN+20} gave an SS-LE protocol for arbitrary graphs.

The second category \cite{BBB13,CP07,FJ06} employs
\emph{oracles}, a kind of failure detectors.
Fischer and Jiang \cite{FJ06}
introduced an oracle $\Omega?$ that eventually
tells all agents whether or not at least one leader exists.
They proposed two SS-LE protocols using $\Omega?$,
one for complete graphs and the other for rings.
Canepa \etal~\cite{CP07} proposed two SS-LE protocols that use $\Omega?$, \ie a deterministic protocol for trees
and a randomized protocol for arbitrary graphs.
Beauquier \etal~\cite{BBB13}
presented a deterministic SS-LE protocol
for arbitrary graphs that uses two copies of $\Omega?$.

The third category \cite{Izu15,SOK+14,SEI+20,SMD+16,SNY+12,SOK+20det,SOK+18,SOK+20polylog} slightly relaxes the requirement of self-stabilization and gave \emph{loosely-stabilizing} leader election protocols. Specifically, the studies of this category allow a population to deviate from the specification of the problem (\ie a unique leader) after the population satisfies the specification for an extremely long time. 
This concept was introduced by \cite{SNY+12}.
The protocols given by \cite{Izu15,SEI+20,SNY+12,SOK+20polylog}
work for complete graphs
and those given by \cite{SOK+14,SMD+16,SOK+20det,SOK+18}
work for arbitrary graphs. 
Recently, Sudo \etal~\cite{SEI+20} gave a time-optimal loosely-stabilizing leader election protocol for complete graphs:
given a design parameter $\tau \ge 1$,
an execution of their protocol reaches a safe configuration
within $O(\tau n\log n)$ expected steps starting from any configuration, and thereafter, it keeps a unique leader
for $\Omega(n^{\tau})$ expected steps,
where $n$ is the number of agents in the population.


The forth category \cite{AAF+08,CC19,CC20,FJ06}
restricts the topology of a graph
to avoid the impossibility of SS-LE.
A class $\calG$ of graphs is called \emph{simple}
if there does not exist a graph in $\calG$ which
contains two disjoint subgraphs that are also in $\calG$.
Angluin \etal~\cite{AAF+08} proves that
there exists no SS-LE protocol
that works for all the graphs in any non-simple class.
Thus, if we focus on a simple class of graphs,
there may exist an SS-LE protocol for all graphs in the class.
Typically, the class of \emph{rings} is simple.
Angluin \etal~\cite{AAF+08} gave an SS-LE protocol that works
for all rings whose sizes are not multiples of a given integer $k$ (in particular, rings of odd size).
They posed a question whether SS-LE is solvable or not for general rings (\ie rings of any size) without any oracle or knowledge such as the exact number of agents in the population,
while Fischer and Jiang \cite{FJ06} solves SS-LE for general rings using oracle $\Omega?$. 
This question had been open for a decade until
Chen and Chen \cite{CC19} recently gave an SS-LE protocol
for general rings.
These three protocols given by \cite{AAF+08,CC19,FJ06}
use only a constant number of states per agent.
The expected convergence times (\ie the expected numbers of steps required to elect a unique leader starting from any configuration)
of the protocols proposed by \cite{AAF+08,FJ06} \footnotemark{}
are $\Theta(n^3)$, while the protocol given by \cite{CC19} requires an exponentially long convergence time.
\footnotetext{
The oracle $\Omega?$ only guarantees that
it \emph{eventually} report to each agent whether
there exists a leader in the population. 
Here, the convergence time of the protocol of \cite{FJ06} is 
bounded by $\Theta(n^3)$ assuming that
the oracle immediately reports the absence of the leader
to each agent.
}
All of the three protocols 
assume that the rings are oriented or \emph{directed}.
However, this assumption is not essential
because Angluin \etal~\cite{AAF+08} also presented 
a self-stabilizing ring orientation protocol,
which gave a common sense of direction to all agents in the ring.
In this paper, we also consider directed rings.
Very recently, Chen and Chen \cite{CC20}
generalized their work on the rings
for regular graphs.

\begin{table*}[t]
 \center
 \caption{Self-Stabilizing Leader Election on Rings}
 \label{tbl:results}
 \begin{tabular}[t]{c c c c}
  \hline
  &Assumption&Convergence Time\ \ & \#states\\
  \hline
 \cite{AAF+08} & $n$ is not multiple of a given $k$
      &$\Theta(n^3)$ & $O(1)$\\
 \cite{FJ06} & oracle $\Omega?$ & $\Theta(n^3)$ & $O(1)$\\
 \cite{CC19} & none & exponential & $O(1)$\\
 this work & \ \ $n \le N$ for a given $N$
& $O(nN)$ & $O(N)$\\
  \hline
 \end{tabular}
\end{table*}

\subsection{Our Contribution}
\label{sec:contribution}

This paper is classified to the fourth category. 
We give the first time-optimal SS-LE protocol $\prl$
for directed rings.
Specifically,
given an upper bound $N$ on the population size $n$,
$\prl$ elects a unique leader
in $O(nN)$ expected steps 
for all directed rings (whose size is at most $N$).
No protocol can solve SS-LE
in $o(n^2)$ expected steps.
Thus, $\prl$ is time-optimal
if a given upper bound $N$ is asymptotically tight,
\ie $N=O(n)$. 
The results are summarized in Table \ref{tbl:results}.

The main and remarkable contribution of this paper is
a novel mechanism that
largely improves the number of steps required
to decrease the number of leaders to one
when there are multiple leaders in the population.
The mechanism requires only $O(n^2)$ expected steps,
while the existing three SS-LE protocols
for rings \cite{AAF+08,CC19,FJ06}
requires $\Omega(n^3)$ expected steps
to decrease the number of leaders to one.
%
%
Our mechanism requires only $O(1)$ states,
which is the same as the existing three protocols. 
(Protocol $\prl$ requires $O(N)$ states only to detect the absence of a leader.)
Thus, if we assume an oracle
that reports to each leader an absence of a leader
within $O(n^2)$ expected steps, we immediately obtain
an SS-LE protocol with $O(n^2)$ expected convergence time
and a constant number of agent-states by using the proposed mechanism to remove leaders.
We leave open an interesting question
whether or not this oracle can be implemented
with $o(N)$ states.

\section{Preliminaries}
In this section, we describe the formal definitions of our computation model.

A \emph{population} is a simple and weakly connected digraph $G(V,E)$,
where $V$ ($|V| \ge 2$) is a set of \emph{agents} and 
$E \subseteq V \times V$ is a set of arcs.
Each arc represents a possible \emph{interactions}
(or communication between two agents): If $(u,v) \in E$,
agents $u$ and $v$ can interact with each other,
where
$u$ serves as an \emph{initiator}
and $v$ serves as a \emph{responder}.
If $(u,v) \notin E$,
agents $u$ and $v$ never have an interaction.
In this paper, we consider only a population represented
by an \emph{directed ring},
\ie we assume that $V=\{u_0, u_1, \dots, u_{n-1}\}$
and $E= \{(u_i,u_{i+1 \bmod n}) \mid i = 0,1,\dots,n-1\}$.
Here, we use the indices of the agents only for simplicity of description. The agents are \emph{anonymous},
\ie they do not have unique identifiers. 
We call $u_{i-1 \bmod n}$ and $u_{i+1 \bmod n}$ the left neighbor and the right neighbor of $u_i$, respectively.
We omit ``modulo by $n$''
(\ie $\bmod~n$)
in the index of agents when no confusion occurs. 
 
A \emph{protocol} $P(Q,Y,T,\outputs)$ consists of 
a finite set $Q$ of states,
a finite set $Y$ of output symbols, 
transition function
$T:  Q \times Q \to Q \times Q$,
and output function $\outputs : Q \to Y$.
When an interaction between two agents occurs,
$T$ determines the next states of the two agents
based on their current states.
The \emph{output of an agent} is determined by $\outputs$:
the output of agent $v$ with state $q \in Q$ is $\outputs(q)$.
We assume that all agents have a common knowledge $N$ on $n$ such that $n \le N = O(\poly(n))$.
Thus, the parameter $Q$, $Y$, $T$, and $\outputs$ 
can depend on the knowledge $N$. 
However, for simplicity,
we do not explicitly write protocol $P$
as parameterized with $N$, \eg 
$P_N = (Q_N, Y_N, T_N, \pi_{\mathrm{out}, N})$.

A \emph{configuration} is a mapping $C : V \to Q$ that specifies
the states of all the agents.
We denote 
the set of all configurations of protocol $P$ by $\call(P)$.
We say that configuration $C$ changes to $C'$ by
an interaction $e = (u_i,u_{i+1})$,
denoted by $C \stackrel{e}{\to} C'$
if we have
$(C'(u_i),C'(u_{i+1}))=T(C(u_i),C(u_{i+1}))$
and $C'(v) = C(v)$
for all $v \in V \setminus \{u_i,u_{i+1}\}$.
We simply write $C \to C'$ if
there exits $e \in E$ such that $C \stackrel{e}{\to} C'$.
We say that a configuration $C'$ is reachable from $C$
if there exists a sequence of configurations
$C_0, C_1, \dots, C_k$ such that $C=C_0$, $C'=C_k$,
and $C_i \to C_{i+1}$ for all $i=0,1,\dots,k-1$.
We also say that a set $\calC$ of configurations
is \emph{closed} if
no configuration outside $\calC$ is reachable from 
a configuration in $\calC$.

A scheduler determines which interaction occurs at each time
step (or just \emph{step}).
In this paper, 
we consider a \emph{uniformly random scheduler}
$\sch=\Gamma_0, \Gamma_1,\dots$: 
each $\Gamma_t \in E$ 
is a random variable
such that
$\Pr(\Gamma_t = (u_i,u_{i+1})) = 1/ n$
for any $t \ge 0$ and $i=0,1,\dots,n-1$.
Each $\Gamma_t$ represents the interaction that occurs
at step $t$.
Given an initial configuration $C_0$,
the \emph{execution} of protocol $P$ under $\sch$
is defined as 
$\Xi_{P}(C_0,\sch) = C_0,C_1,\dots$ such that
$C_t \stackrel{\Gamma_t}{\to} C_{t+1}$ for all $t \ge 0$. 
We denote $\Xi_{P}(C_0,\sch)$
simply by $\Xi_{P}(C_0)$ when no confusion occurs.

We address the self-stabilizing leader election problem (SS-LE)
in this paper. For simplicity,
we give the definition of a self-stabilizing leader election protocol instead of giving the definitions of self-stabilization and the leader election problem separately.
\begin{dfn}[Self-stabilizing Leader Election]
\label{def:ss-le} 
For any protocol $P$,
we say that a configuration $C$ of $P$ is \emph{safe}
if (i) exactly one agent outputs $L$
and all other agents output $F$ in $C$,
and
(ii) 
at every configuration reachable from $C$,
all agents keep the same outputs as those in $C$.
A protocol $P$ is a \emph{self-stabilizing leader election protocol} if  $\Xi_{P}(C_0,\sch)$ reaches a safe configuration
with probability $1$.
\end{dfn}

We evaluate a self-stabilizing leader election protocol $P$
with two metrics:
\emph{the expected convergence time}
and \emph{the number of states}.
For a configuration $C \in \call(P)$,
let $t_{P,C}$ be the expected number of steps 
until $\Xi_{P}(C_0,\sch)$ reaches a safe configuration.
The expected convergence time of $P$ is
defined as $\max_{c \in \call(P)} t_{P,C}$.
The number of states of $P=(Q,Y,T,O)$ is simply $|Q|$.

\section{Self-Stabilizing Leader Election Protocol}
In this section, we propose a self-stabilizing leader election (SS-LE) protocol $\prl$ that works in any directed ring consisting of less than or equal to $N$ agents.
The expected convergence time is $O(nN)$,
and the number of states is $O(N)$.
Thus, $\prl$ is time-optimal if
a given upper bound $N$ of $n$ is asymptotically tight, \ie $N=O(n)$.

The pseudocode of $\prl$ is given in Algorithm \ref{al:prl},
which describes how two agents $l$ and $r$ updates their states, \ie their variables, when they have an interaction.
Here, $l$ and $r$ represents the initiator
and the responder in the interaction, respectively.
That is, $l$ is the left neighbor of $r$,
and  $r$ is the right neighbor of $l$.
We denote the value of the variable $\mathtt{var}$
at agent $v\in V$ by $v.\mathtt{var}$.
Similarly, we denote the variable $\mathtt{var}$
in state $q\in Q$ by $q.\mathtt{var}$.
In this algorithm, each agent $v \in V$ maintains
an \emph{output} variable $v.\leader \in \{0,1\}$,
according to which it determines its output.
Agent $v$ outputs $L$ when $v.\leader = 1$
and outputs $F$ when $v.\leader = 0$.
We say that $v$ is a \emph{leader} if $v.\leader = 1$;
otherwise $v$ is called a \emph{follower}.
For each $u_i \in V$,
we define
the distance to \emph{the nearest left leader} 
and the distance to \emph{the nearest right leader}
of $u_i$
as 
$\distl(i) = \min \{j \ge 0 \mid u_{i-j}.\leader = 1\}$
and
$\distr(i) = \min \{j \ge 0 \mid u_{i+j}.\leader = 1\}$,
respectively.
When there is no leader in the ring, 
we define $\distl(i)=\distr(i)=\infty$.

Algorithm $\prl$ consists of two parts:
the leader creation part (Lines 1--5)
and the leader elimination part (Lines 6--19).
Since $\prl$ is a self-stabilizing protocol,
it has to handle any initial configuration,
where there may be no leader
or multiple leaders. 
The leader creation part creates a new leader
when there is no leader,
while the leader elimination part 
decreases the number of leaders
to one when there are two or more leaders. 

\subsection{Leader Elimination}
\label{sec:elimination}
We are inspired by the algorithm of \cite{FJ06} to design the leader elimination part (Lines 6--19).
Roughly speaking,
the strategy of \cite{FJ06} can be described as follows.
\begin{itemize}
 \item Each agent may have a \emph{bullet} and/or a \emph{shield}.
 \item A leader always fire a bullet:
each time a leader $u_{i+1}$ having no bullet
interacts with an agent $u_{i}$,
the leader $u_{i+1}$ makes a bullet.
 \item Bullets move from left to right in the ring:
each time $u_i$ having a bullet interacts with $u_{i+1}$,
the bullet moves from $u_i$ to $u_{i+1}$.
 \item Conversely, shields move from right to left:
each time $u_{i+1}$ having a shield interacts with $u_{i}$,
the shield moves from $u_{i+1}$ to $u_{i}$.
 \item Each time two agents both with bullets (resp.~shields)
have an interaction,
the left bullet (resp.~the right shield) disappears.
 \item When a bullet and a shield pass each other,
\ie $u_i$ with a bullet and $u_{i+1}$ with a shield
have an interaction, the bullet disappears.
 \item When a bullet moves to a leader,
the leader is \emph{killed} (\ie becomes a follower). 
\end{itemize}
The algorithm of \cite{FJ06} assumes an oracle,
called an \emph{eventual leader detector} $\Omega?$,
which
detects and tells each agent whether
a leader exists or not, 
when there is continuously a leader or
there is continuously no leader. 
A follower becomes a leader when it is reported by $\Omega?$
that there is no leader in the population. 
At this time, the new leader simultaneously
generates both a shield and a bullet.
One can easily observe that 
by the above strategy together with oracle $\Omega?$,
the population eventually reaches a configuration
after which there is always one fixed leader.
However, the algorithm of \cite{FJ06}
requires $\Omega(n^3)$ steps to elect one leader
in the worst case
even if oracle $\Omega?$ can immediately report to each agent
whether there is a leader in the population.

\begin{algorithm}[t]
\caption{$\prl$}
\label{al:prl}


\textbf{Interaction} between initiator $l$ and responder $r$:
\begin{algorithmic}[1]
\State $l.\dist \gets \begin{cases}
		 0 & l.\leader = 1\\		 
		 l.\dist  & \text{otherwise}
		\end{cases}$
\State $r.\dist \gets \begin{cases}
		 0 & r.\leader = 1\\
		 \min(l.\dist + 1,N) & r.\leader = r.\bull = 0\\
		 r.\dist  & \text{otherwise}
		\end{cases}$
\If{$r.\dist = N$}
\State $r.\leader \gets 1$; $r.\bull \gets 2$;  $r.\shield \gets 1$; $r.\signal \gets 0$; $r.\dist \gets 0$; 
\EndIf
\vblank

\If{$l.\leader = l.\signal = 1$}
\State $l.\bull \gets 2$; $l.\shield \gets 1$; $l.\signal \gets 0$
\EndIf

\If{$r.\leader = r.\signal = 1$}
\State $r.\bull \gets 1$; $r.\shield \gets 0$; $r.\signal \gets 0$
\EndIf

\If{$l.\bull > 0 \land r.\leader = 1$}
\State $r.\leader \gets \begin{cases}
		     0 & l.\bull = 2 \land r.\shield = 0 \\
		     1 & \text{otherwise}
		    \end{cases}$
\State $l.\bull \gets 0$; 
\ElsIf{$l.\bull > 0 \land r.\leader = 0$}
\State $r.\bull \gets \begin{cases}
		      l.\bull & r.\bull = 0\\		       
		      r.\bull & r.\bull > 0
		      \end{cases}$
\State $l.\bull \gets 0$;
$r.\signal \gets 0$;
\EndIf
\State $l.\signal \gets \max(l.\signal,r.\signal,r.\leader)$
%
%
\end{algorithmic}
\end{algorithm}

We drastically modify the above strategy of \cite{FJ06}
for the leader elimination part of $\prl$
to decrease the number of leaders to one within $O(n^2)$ steps.
First, a shield never moves in our algorithm.
Only leaders have a shield.
A leader sometimes generates a shield
and sometimes breaks a shield.
Second, a leader does not always fire a bullet.
Instead, a leader fires a new bullet 
only after it detects that the last bullet it fired
reaches a (possibly different) leader.
Third, we have two kinds of bullets:
\emph{live bullets} and \emph{dummy bullets}.
A live bullet kills a leader without a shield.
However, a dummy bullet does not have capability
to kill a leader. 
When a leader decides to fire a new bullet,
it fires a live bullet with probability 1/2
and fires a dummy bullet with the rest probability.
When a leader fires a live bullet, 
it simultaneously generates a shield (if it does not have a shield). 
When a leader fires a dummy bullet,
it breaks the shield if it has a shield. 
Thus, roughly speaking,
each leader is \emph{shielded} (\ie has a shield)
with probability $1/2$ at each step.
Therefore, when a live bullet reaches a leader,
the leader is killed with probability $1/2$.
This strategy is well designed: not all leaders kill each other
simultaneously because a leader must be shielded if
it fired a live bullet in the last shot. 
As a result, the number of leaders eventually decreases to one.

In what follows, we explain how we implement this strategy.
Each agent $v$ maintains variables $v.\bull \in \{0,1,2\}$, $v.\shield \in \{0,1\}$, and $v.\signal \in \{0,1\}$.
As their names imply, $v.\bull=0$ (resp.~$v.\bull=1$, $v.\bull=2$)
indicates that $v$ is now conveying no bullet (resp.~a dummy bullet, a live bullet), while $v.\shield=1$ indicates that $v$ is shielded.
Unlike the protocol of \cite{FJ06},
we ignore the value of $v.\shield$ for any follower $v$.
A variable $\signal$ is used by a leader to detect
that the last bullet it fired already disappeared.
Specifically, $v.\signal=1$ indicates that 
$v$ is propagating a \emph{bullet-absence signal}.
A leader always generates a bullet-absence signal
in its left neighbor when it interacts with its
left neighbor (Line 19).
This signal propagates from right to left (Line 19),
while a bullet moves from left to right (Lines 16-17).
A bullet disables a bullet-absence signal
regardless of whether it is live or dummy,
\ie $u_{i+1}.\signal$ is reset to $0$
when two agents $u_i$ and $u_{i+1}$ such that $u_i.\bull > 0$
and $u_{i+1}.\signal=1$ have an interaction (Lines 16 and 17).
Thus, a bullet-absence signal propagates to a leader only after the last bullet fired by the leader disappears. 

When a leader $u_i$ receives a bullet-absence signal from its right neighbor $u_{i+1}$, $u_i$ waits for its next interaction
to extract randomness from the uniformly random scheduler.
At the next interaction, by the definition of the uniformly random scheduler, $u_{i}$ meets its right neighbor $u_{i+1}$ with probability $1/2$ and its left neighbor $u_{i-1}$ with probability $1/2$. In the former case, $u_i$ fires a live bullet and becomes shielded (Lines 6--8). In the latter case, $u_i$ fires a dummy bullet
and becomes unshielded (Lines 9--11).
In both cases, the received signal is deleted (Lines 7 and 10).
The fired bullet moves from left to right each time
the agent with the bullet, say $u_i$, interact
with its right neighbor $u_{i+1}$ (Lines 16 and 17).
However, the bullet disappears without moving to $u_{i+1}$
if $u_{i+1}$ already has another bullet at this time.
Suppose that the bullet now reaches a leader.
If the bullet is live and the leader is not shielded
at that time, the leader is killed by the bullet (Line 13).
The bullet disappears at this time regardless of whether the bullet is alive and/or the leader is shielded (Line 14). 

\subsection{Leader Creation}
\label{sec:creation}

The leader creation part is simple  (Lines 1--5).
Each agent $u_i$ estimates $\distl(i)$
and stores the estimated value on variable $u_i.\dist \in \{0,1,\dots,N\}$.
Specifically,
at each interaction $(u_i,u_{i+1})$, 
agents $u_i$ and $u_{i+1}$ updates their $\dist$ as follows
(Lines 1 and 2):
(i) $u_i$ (resp.~$u_{i+1}$) resets its $\dist$ to zero
if $u_i$ (resp.~$u_{i+1}$) is a leader,
and 
(ii) if $u_{i+1}$ is not a leader and
does not have a bullet,
$\min(l.\dist+1,N)$ is substituted for $u_{i+1}.\dist$.
Thus, if there is no leader in the population,
some agent $v$ eventually increases $v.\dist$ to $N$,
and at that time, the agent decides that there is no leader.
Then, this agent becomes a leader, 
executing $v.\leader \gets 1$ and $v.\dist \gets 0$ (Line 4).
At the same time, 
$v$ fires a live bullet, generates a shield,
and disables a bullet-absence signal (Line 4).
This live bullet prevents the new leader from being killed
for a while:
the leader becomes unshielded only after it receives a bullet-absence signal, and the live bullet prevent the leader receives
a bullet-absence signal before another shielded leader appears.

As mentioned above, at interaction $(u_i,u_{i+1})$,
the distance propagation does not occur if
$u_{i+1}$ is a leader. 
This exception helps us to simplify the analysis of the convergence time, \ie we can easily get an upper bound on the expected number of steps before each bullet disappears.
Note that
there are two cases that a bullet disappears:
(i) when it reaches a leader,
and 
(ii) when it reaches another bullet.
The first case includes an interaction $(u_i,u_{i+1})$
where $u_{i}.\dist \ge N-1$ holds
and $u_{i+1}$ becomes a leader by Line 4.
Formally,
at an interaction $(u_i,u_{i+1})$ such that $u_i.\bull \ge 1$,
we say that a bullet located at $u_i$
disappears if $u_{i+1}.\leader = 1$, $u_{i+1}.\bull \ge 1$,
or $u_i.\dist \ge N-1$. 
We have the following lemma thanks to the above exception.

\begin{lemma}
\label{lem:disappear} 
Every bullet disappears before it moves (from left to right)
$N$ times.
\end{lemma}


We should note that the leader creation part may create a leader even when there is one or more leaders, thus this part may prevent
the leader elimination part from decreasing the number of leaders
to one.
Fortunately, as we will see in Section \ref{sec:correctness},
within $O(nN)$ steps in expectation,
the population reaches a configuration after which
no new leader is created.

\section{Correctness and Time Complexity}
\label{sec:correctness}
In this section, we prove
that $\prl$ is a self-stabilizing leader election protocol
on directed rings of any size
and that the expected convergence time of $\prl$
is $O(nN)$.
In Section \ref{sec:srl},
we define a set $\srl$ of configurations
and prove that every configuration in
$\srl$ is safe.
In Section \ref{sec:convergence},
we prove that
the population starting from any configuration reaches 
a configuration in $\srl$ within $O(nN)$ steps in expectation.


\subsection{Safe Configurations}
\label{sec:srl}


In this paper,
we use several functions whose return values
depend on a configuration,
such as $\distl(i)$ and $\distr(i)$.
When a configuration should be specified,
we explicitly write a configuration as the first argument
of those functions.
For example, 
we write $\distl(C,i)$ and $\distr(C,i)$
to denote $\distl(i)$ and $\distr(i)$
in a configuration $C$, respectively.

In protocol $\prl$, leaders kill each other by firing live bullets
to decrease the number of leaders to one.
However, it is undesirable that
all leaders are killed and the number of leaders becomes zero. 
Therefore, a live bullet should not kill a leader if it is the last leader (\ie the unique leader) in the population.
We say that a live bullet located at agent $u_i$ is \emph{peaceful} when the following predicate holds:
\begin{align*}
\peaceful(i) &\predef
\left(
\begin{aligned}
&u_{i-\distl(i)}.\shield = 1 \\
& \land \forall j = 0,1,\dots, \distl(i): u_{i-j}.\signal = 0
\end{aligned}
\right).
\end{align*}
A peaceful bullet never kills the last leader in the population
because its nearest left leader is shielded.
A peaceful bullet never becomes non-peaceful;
because letting $u_i$ be the agent at which the bullet is located,
the agents
$u_{i-\distl(i)}, u_{i-\distl(i)+1},\allowbreak \dots,u_{i}$
will never have a bullet-absence signal
thus $u_{i-\distl(i)}$ never becomes unshielded
before the bullet disappears.
At the beginning of an execution, 
there may be one or more non-peaceful live bullets. 
However, every newly-fired live bullet is peaceful because
a leader becomes shielded and disables  the bullet-absence signal
when it fires a live bullet. 
Thus, once the population reaches a configuration where
every live bullet is peaceful and there is one or more leaders,
the number of leaders never becomes zero.
Formally, we define the set of such configurations
as follows:
\begin{align*}
 \cpb &= \left \{ C \in \call(\prl) \longmid
\begin{aligned}
&\exists u_i \in V: C(u_i).\leader = 1\\
&\land \forall u_j \in V: 
C(u_j).\bull = 2 \Rightarrow \peaceful(C,j) 
\end{aligned}
\right  \}.
\end{align*}
The following lemma holds from the above discussion.
\begin{lemma}
\label{lem:cpb_closed} 
$\cpb$ is closed.
\end{lemma}
\noindent
Thus, once the population reaches a configuration in $\cpb$,
there is always one or more leaders.

In protocol $\prl$,
a new leader is created when $\dist$ of some agent
reaches $N$.
We require this mechanism to create a new leader
when there is no leader. 
However, it is undesirable that
a new leader is created when there is already one or more leaders.
We say that an agent $u_i$ is \emph{secure}
when the following predicate holds:
\begin{align*}
\secure(i) &\predef
\begin{cases}
u_i.\dist = 0 & u_i.\leader = 1\\ 
u_i.\dist \le N - \distr(i) & \text{otherwise}
\end{cases}.
\end{align*}
One may think that no leader is created 
once the population reaches a configuration in $\cpb$
such that all agents are secure.
Unfortunately, this does not hold.
For example, consider the case $n=N=100$
and a configuration $C \in \cpb$
where
\begin{itemize}
 \item only two agents $u_{0}$ and $u_{50}$ are leaders,
 \item $u_{0}.\dist = u_{50}.\dist = 0$, 
 \item $u_{i}.\dist = 100-\distr(i)$ for all
$i = 1,2,\dots,49,51,52,\dots,100$, 
 \item $u_{49}$ carries a live bullet in $C$,
\ie $u_{49}.\bull = 2$, and
 \item $u_{50}$ is not shielded, \ie $u_{50}.\shield = 0$.
\end{itemize}
Note that the above condition does not contradict
the assumption $C \in \cpb$.
In this configuration, all agents are secure. 
However, starting from this configuration,
the population may create a new leader
even when another leader exists.
In configuration $C$, $u_{49}.\dist = 99$.
If $u_{49}$ and $u_{50}$ have two interactions in a row, 
then $u_{50}$ becomes a follower in the first interaction,
and $u_{49}.\dist + 1 = 100$ is substituted for $u_{50}.\dist$
and $u_{50}$ becomes a leader again in the second interaction
(even though $u_0$ is a leader during this period).

We introduce the definition of \emph{modest bullets}
to clarify the condition by which
a new leader is no longer created.
A live bullet located at $u_i$ is said to be \emph{modest}
when the following predicate holds:
\begin{align*}
\modest(i) &\predef
\peaceful(i) \land
 \forall j=0,1,\dots,\distl(i): u_{i-j}.\dist \le \distl(i-j). 
\end{align*}
As we will see soon,
a new leader is no longer created
in an execution starting from a configuration in $\cpb$ where
all agents are secure and 
all live bullets are modest.
Note that 
in the above example, 
a live bullet located at $u_{49}$ in $C$
is not modest. 
We define a set $\cni$ of configurations as follows:
\begin{align*}
\cni &= \left \{ C \in \cpb \longmid
\begin{aligned}
&\forall u_i \in V: \secure(C,i) \\
&\land \forall u_j \in V: 
C(u_j).\bull = 2 \Rightarrow \modest(C,j)
\end{aligned}
\right \}.
\end{align*}


\begin{lemma}
\label{lem:modest} 
A modest bullet never becomes non-modest.
\end{lemma}

\begin{proof}
Let 
$C$ and $C'$ be any configurations
such that $C \to C'$.
and $b$ be a modest bullet in $C$.
Assume for contradiction that 
$b$ does not disappear
and $b$ becomes non-modest in $C \to C'$.
Let $u_i$ and $u_{i'}$
be the agents at which $b$ is located in $C$ and $C'$,
respectively ($i' \in \{i,i+1\}$).
Since a peaceful bullet never becomes non-peaceful,
$b$ is still peaceful in $C'$.  
By definition of a modest bullet, in $C$,
agent $u_{i-j}$ satisfies $u_{i-j}.\dist \le \distl(i-j)$ 
for all $j=0,1,\dots,\distl(C,i)$.
Thus, none of $u_{i},u_{i-1},\dots,u_{i-\distl(C,i)+1}$
becomes a leader in $C \to C'$,
Moreover,
$u_{i-\distl(C,i)}$ is shielded in $C$
and thus never becomes a follower in $C \to C'$.
This yields that the nearest left leader of $b$ 
does not change in $C \to C'$,
\ie $i-\distl(C,i)=i'-\distl(C',i')$.
Therefore, for all $j=0,1,\dots,\distl(C,i)$,
$u_{i-j}$ still satisfies $u_{i-j}.\dist \le \distl(i-j)$ in $C'$.
Since $b$ is not modest in $C'$,
$b$ must move $u_i$ to $u_{i+1}$ in $C \to C'$,
and $u_{i+1}.\dist > \distl(C',i+1)$ must hold in $C'$.
However, in $C \to C'$,
$u_{i+1}.\dist$ is updated to
$u_{i}.\dist + 1 \le \distl(C,i) +1 = \distl(C',i+1)$,
a contradiction.
\end{proof}

\begin{lemma}
\label{lem:newbullet} 
A newly-fired live bullet is modest.
\end{lemma}
\begin{proof}
Assume that a leader $u_i$
fires a live bullet $b$ at interaction $(u_i,u_{i+1})$
in $C \to C'$.
Bullet $b$ immediately disappears by Lines 14 if $u_{i+1}$
is a leader or has a bullet in $C$.
Otherwise, $b$ moves to $u_{i+1}$.
Then, $u_i.\shield = 1$,
$u_i.\dist = 0$, $u_{i+1}.\dist=0$,
and $u_i.\signal = u_{i+1}.\signal = 0$
must hold in $C'$,
which yields that $b$ is modest in $C'$.
\end{proof}

\begin{lemma}
\label{lem:modest_secure} 
Let $C$ be any configuration where all live bullets are modest
and $C'$ any configuration such that $C \to C'$.
Then, 
a secure agent $u_i$ becomes insecure in $C \to C'$ only if
$u_i$ interacts with $u_{i-1}$ in $C \to C'$
and $u_{i-1}$ is insecure in $C$.
\end{lemma}

\begin{proof}
Assume that $u_i$ becomes insecure in $C \to C'$.
First, consider the case that
no leader becomes a follower in $C \to C'$.
Then, $u_i$ must change the value of $u_i.\dist$
to become insecure.
If $u_i$ interacts with $u_{i+1}$, 
$u_i.\dist$ does not change or becomes zero.
Thus, $u_i$ must interact with $u_{i-1}$
and increase $u_{i}.\dist$ to
a value greater than $N-\distr(C',i) \ge N-\distr(C,i-1)+1$,
hence $u_{i-1}$ must be insecure in $C$.
Next, consider the case that
a leader $u_j$ becomes a follower in $C \to C'$.
If $u_{j} \neq u_{i + \distr(C,i)}$, 
$\distr(C,i)=\distr(C',i)$ 
and $C(u_i).\dist = C'(u_i).\dist$ hold,
which violates the assumption that $u_i$ becomes insecure.
Thus, $u_j = u_{i + \distr(C,i)}$ holds.
However, this means that
a modest bullet is located at
$u_{i}, u_{i+1}, \dots, u_{j-1}$ in $C$.
Thus, by definition of a modest bullet, 
$C'(u_i).\dist \le C(u_i).\dist \le \distl(C,i) \le N-\distr(C',i)$ holds. Hence, $u_i$ is still secure in $C'$,
which violates the assumption.
To conclude, 
$u_i$ must interact with $u_{i-1}$ in $C \to C'$
and $u_{i-1}$ must be insecure in $C$.
\end{proof}

\begin{lemma}
\label{lem:cni_closed} 
$\cni$ is closed.
\end{lemma}

\begin{proof}
Immediately follows from
Lemmas \ref{lem:cpb_closed}, \ref{lem:modest},
\ref{lem:newbullet}, and \ref{lem:modest_secure}.
\end{proof}

\begin{lemma}
\label{lem:noleader} 
No new leader is created
in an execution starting from a configuration in $\cni$.
\end{lemma}

\begin{proof}
Since $\cni$ is closed by Lemma \ref{lem:cni_closed},
every configuration that appears in an execution
starting from a configuration in $\cni$
is also in $\cni$.
All agents are secure in a configuration in $\cni$.
Thus, $u_i.\dist \le N-1$ holds for all $u_i \in V$
and $u_i.\dist = N-1$ holds only if
$u_{i+1}$ is a leader.
Therefore, in a configuration in $\cni$,
no agent increases its $\dist$ to $N$,
thus no new leader is created.
\end{proof}

Finally,
we define $\srl$ as the set of all configurations
included in $\cni$
where there is exactly one leader. 

\begin{lemma}
\label{lem:srl_closed} 
$\srl$ is closed and includes only safe configurations.
\end{lemma}

\begin{proof}
Let $C$ be any configuration in $\srl$
and $C'$ any configuration such that $C \to C'$.
Since $\cni$ is closed by Lemma \ref{lem:cni_closed}
and exactly one agent
is a leader in $C$,
it suffices to show that
no one changes its output (\ie the value of variable $\leader$)
in $C \to C'$.
Since $C \in \cpb$,
the unique leader in $C$ is never killed
in $C \to C'$.
By Lemma \ref{lem:noleader},
any other agent becomes a leader
in $C \to C'$.
Thus, no agent changes its output 
in $C \to C'$.
\end{proof}

\subsection{Convergence}
\label{sec:convergence}
In this subsection, we prove
that an execution of $\prl$ 
starting from any configuration in $\call(\prl)$
reaches a configuration in $\srl$
within $O(nN)$ steps in expectation.
Formally,
for any $C \in \call(\prl)$ and $\calS \subseteq \call(\prl)$,
we define $\ect(C,\calS)$ as the expected number of steps
that execution $\Xi_{\prl}(C,\sch)$ requires to reach 
a configuration in $\calS$.
The goal of this subsection is
to prove $\max_{C \in \call(\prl)} \ect(C,\srl) \allowbreak =O(nN)$.
We give this upper bound by showing
$\max_{C \in \call(\prl)} \ect(C,\cni) =O(nN)$
and 
$\max_{C \in \call(\cni)} \ect(C,\srl)=O(n^2)$
in Lemmas 
\ref{lem:cni} and \ref{lem:srl}, respectively.



In this subsection,
we denote interaction $(u_i,u_{i+1})$ by $e_i$.
In addition,
for any two sequences of interactions
$s=e_{k_0},e_{k_1},\dots,e_{k_h}$
and $s'=e_{k'_0},e_{k'_1},\dots,e_{k'_{j}}$,
we define
$s\cdot s' = e_{k_0},e_{k_1},\dots,e_{k_h},e_{k'_0},e_{k'_1},\dots,e_{k'_{j}}$. That is, we use ``$\cdot$'' for the concatenation operator.
For any sequence $s$ of interactions and integer $i \ge 1$,
we define $s^i$ by induction:
$s^1 = s$ and $s^i = s\cdot s^{i-1}$.
For any $i,j \in \{0,1,\dots,n-1\}$,
we define $\seqr(i,j)=e_i,e_{i+1},\dots,e_j$
and $\seql(i,j)=e_i,e_{i-1},\dots,e_{j}$.

\begin{dfn}
Let $\gamma=e_{k_1},e_{k_2},\dots,e_{k_h}$
be a sequence of interactions.
We say that $\gamma$ occurs within $l$ steps when
$e_{k_1},e_{k_2},\dots,e_{k_h}$ occurs 
in this order (not necessarily in a row) 
within $l$ steps.
Formally,
the event
``$\gamma$ occurs within $l$ steps from a time step $t$''
is defined as the following event:
$\Gamma_{t_i}=e_{k_i}$ holds for all $i = 1,2,\dots,h$
for some sequence of integers 
$t \le t_1 < t_2 < \dots < t_{h} \le t+l-1$.
We say that from step $t$,
$\gamma$ completes at step $t+l$
if $\gamma$ occurs within $l$ steps
but does not occur within $l-1$ steps.
%
When $t$ is clear from the context, 
we write ``$\gamma$ occurs within $l$ steps''
and ``$\gamma$ completes at step $l$'',
for simplicity.
\end{dfn}


\begin{lemma}
\label{lem:occurs}
From any time step,
a sequence $\gamma=e_{k_1},e_{k_2},\dots,e_{k_h}$
with length $h$
occurs within $nh$ steps in expectation.
\end{lemma}
\begin{proof}
For any interaction $e_i$,
at each step, $e_i$ occurs with probability $1/n$.
Thus, $e_i$ occurs within $n$ steps in expectation.
Therefore, $\gamma$ occurs within $nh$ steps in expectation.
\end{proof}

\begin{lemma}
\label{lem:create} 
Let $C$ be a configuration where no leader exists.
In execution $\Xi = \Xi_{\prl}(C,\sch)$,
a leader is created within $O(nN)$ steps in expectation.
\end{lemma}

\begin{proof}
Let $\gamma = (\seqr(0,n-1))^{\lceil N/n \rceil+1}$.
Since the length of $\gamma$ is at most $N+2n$,
$\gamma$ occurs within $3nN$ steps in expectation by Lemma \ref{lem:occurs}.
Thus, it suffices to show that
a leader is created before $\gamma$ completes in $\Xi$.
Assume for contradiction that 
no leader is created before $\gamma$ completes in $\Xi$.

First, consider the case that there is at least one bullet
in $C$. 
Let $B$ be the set of bullets that exist in $C$.
Before $\gamma$ completes, all bullets in $B$ disappear
by Lemma \ref{lem:disappear}.
This implies that
at least one bullet in $B$ disappear
by reaching a bullet not included in  $B$,
since there is no leader during the period.
However, no bullet $b \notin B$ exists during the period
because a follower never creates a new bullet.
This is a contradiction.

Second, consider the case that there is no bullet in $C$.
Let $u_i$ be any agent with the minimum $\dist$ in $C$.
Each time $\seqr(i,i-1)$ completes,
$u_i.\dist$ increases at least by $n$
unless a new leader is created. 
Since $(\seqr(i,i-1))^{\lceil N/n \rceil}$ completes earlier than
$\gamma = (\seqr(0,n-1))^{\lceil N/n \rceil+1}$,
a new leader is created before $\gamma$ completes,
a contradiction.
\end{proof}

\begin{lemma}
\label{lem:cni} 
$\max_{C \in \call(\prl)} \ect(C,\cni)=O(nN)$.
\end{lemma}

\begin{proof}
Let $C_0$ be any configuration in $\call(\prl)$
and consider $\Xi = \Xi_{\prl}(C_0,\sch)$.
All leaders that exist in $C_0$ disappear
before $\gamma = (\seqr(0,n-1))^{\lceil N/n \rceil+1}$ completes
by Lemma \ref{lem:disappear},
while $\gamma$ occurs
within $O(nN)$ steps in expectation by Lemma \ref{lem:occurs}.
Thus, by Lemmas \ref{lem:modest}, \ref{lem:newbullet}, \ref{lem:create}, within $O(nN)$ steps in expectation,
$\Xi$ reaches a configuration $C'$
where all live bullets are modest, there is at least one leader,
and every leader $u_i$ satisfies $u_i.\dist = 0$.

Let $\Xi'$ be the suffix of $\Xi$ after $\Xi$ reaches $C'$.
In the rest of this proof,
we show that $\Xi'$ reaches a configuration
in $\cni$ within $O(n^2)$ steps. 
Since $\cpb$ is closed (Lemma \ref{lem:cpb_closed})
and $C' \in \cpb$,
there is always at least one leader in $\Xi'$.
Thus, by Lemmas \ref{lem:modest} and \ref{lem:newbullet},
it suffices to show that
$\Xi'$ reaches a configuration where all agents are secure
within $O(n^2)$ steps in expectation.

Here, we have the following two claims.
\begin{clm}
\label{clm:secure}
In $\Xi'$, once an agent $u_i$ becomes a leader,
$u_i$ is always secure thereafter
(even after it becomes a follower).
\end{clm}
\begin{proof}
Suppose that $u_i$ is a leader in some point of $\Xi'$. 
At this time, $u_i.\dist = 0$.
As long as $u_i$ is a leader, $u_i$ is secure.
Agent $u_i$ becomes a follower only
when a live bullet reaches $u_i$.
In $\Xi'$, all live bullets are modest.
This yields that, when $u_i$ becomes a follower,
all agents $u_{i},u_{i-1},\dots,u_{i-\distl(i)}$
are secure.
Thus, letting $u_j = u_{i-\distl(i)}$,
agent $u_{i}$ is secure as long as $u_{j}$ is secure.
Similarly, $u_j$ is secure as long as $u_j$ is a leader.
Even if $u_j$ becomes a follower,
there is a leader $u_k$ such that
$u_j$ is secure as long as $u_k$ is a leader, and so on.
Therefore, $u_i$ never becomes insecure.
\end{proof}

\begin{clm}
\label{clm:become_secure}
In $\Xi'$, 
an insecure agent $u_i$ becomes secure
if it interacts with $u_{i-1}$ when $u_{i-1}$ is secure. 
\end{clm}
\begin{proof}
Let $C \to C'$ be any transition that appears in $\Xi$
such that $u_{i-1}$ is secure in $C$.
If $u_{i-1}$ is a leader in $C$ (thus in $C'$),
$C'(u_i).\dist \le 1 \le N-\distr(C',i)$.
Otherwise, $\distr(C',i)=\distr(C',i-1)+1$ must hold.
By Lemma \ref{lem:modest_secure},
$u_{i-1}$ is still secure in $C'$,
hence $C'(u_i).\dist \le C'(u_{i-1}).\dist + 1 \le N -\distr(C',i-1) + 1 = N-\distr(C',i)$.
Thus, $u_{i}$ is secure in $C'$ in both cases.
\end{proof}

Let $u_i$ be a leader in $C'$.
By Lemma \ref{lem:modest_secure}
and Claims \ref{clm:secure} and \ref{clm:become_secure},
all agents are secure when $\seqr(i,i-2)$ completes.
This requires $O(n^2)$ expected steps by Lemma \ref{lem:occurs}.
\end{proof}


\begin{lemma}
\label{lem:killed}
Let $C_0$ be a configuration in $\cni$
where there are at least two leaders.
Let $u_i$, $u_j$, and $u_k$ be the leaders in $C_0$ 
such that $i = j - \distl(C_0,j)$ and $j = k - \distl(C_0,k)$,
\ie $u_i$ is the nearest left leader of $u_j$
and $u_j$ is the nearest left leader of $u_k$ in $C_0$.
($u_i = u_k$ may hold.)
Let $d_1=\distl(C_0,j)$ and $d_2 = \distl(C_0,k)$.
Then, in an execution $\Xi=\Xi_{\prl}(C_0,\sch)=C_0,C_1,\dots$,
the event that one of the three leaders becomes a follower
occurs within $O(n(d_1 + d_2))$ steps in expectation.
\end{lemma}


\begin{proof}
 Let $\tdel$ be the minimum integer $t$ such that
 $u_i$, $u_j$, or $u_k$ is a follower in $C_{t}$.
 Our goal is to prove $\ex[\tdel] = O(n(d_1+d_2))$.
 Let $V' = \{u_i,u_{i+1},\dots,u_k\}$.
 Consider the protocol, denoted by $\pdummy$,
 that can be obtained by removing Line 13 from $\prl$.
 No leader becomes a follower
 in any execution of $\pdummy$.
 Let $\Xi'=\Xi_{\pdummy}(C_0,\sch)=D_0,D_1,\dots$.
 Of course, $\Xi$ and $\Xi'$ can be different.
 However, each agent in $V'$ always has the same state
 both in $\Xi$ and $\Xi'$
 before $u_i$, $u_j$, or $u_k$ becomes a follower in $\Xi$. 
 Formally, $C_t(u_s)=D_t(u_s)$ holds for all $u_s \in V'$
 and all $t=0,1,\dots,\tdel-1$.
 Therefore, we have $\tdel \le \tdel'$,
 where $\tdel'$ is the minimum integer $t$
 such that a live bullet reaches $u_j$ in $D_{t-1} \to D_t$
 and $u_j$ is unshielded at this interaction,
 \ie $D_{t-1}(u_j).\shield=0$;
 because if $u_j$ is still a leader in $C_t$,
 $u_i$ or $u_k$ must become a follower in $C_0,C_1,\dots,C_{t-1}$.
 Thus, it suffices to show that $\tdel'=O(n(d_1+d_2))$.
 
Again, no leader becomes a follower in $\Xi'$.
In addition, no leader is created in $\Xi'$
by Lemma \ref{lem:noleader}.
Leader $u_j$ fires a bullet at least once in $\Xi'$ 
before or when $\gamma = \seqr(j,k-1)\cdot \seql(k-1,j) \cdot e_j$ completes.
Thereafter, at any step, 
$u_j$ is shielded with probability at least $1/2$
because
(i) each time $u_j$ fires a bullet,
 it fires a live bullet and becomes shielded with probability $1/2$ and fires a dummy bullet and becomes
unshielded with probability $1/2$, and
(ii) live bullets fired by $u_j$ reach $u_k$
not later than dummy bullets fired by $u_j$
because live bullets are initially located at $u_{j+1}$
when they are fired by $u_j$, while
dummy bullets are initially located at $u_j$ when they are fired
by $u_j$.
After $u_j$ fires a bullet for the first time, 
$u_i$ fires a bullet at least once
and the bullet reaches $v_j$ in $\Xi'$
before or when
$\gamma' = \seqr(i,j-1)\cdot \seql(j-1,i) \cdot \seqr(i,j-1)$ completes.
This bullet is a live one with probability $1/2$,
while $u_j$ is unshielded at this time
with probability at least $1/2$, as mentioned above.
Thus, each time $\gamma \cdot \gamma'$ completes,
the event that a live bullet reaches $r_j$ at the time
$r_j$ is shielded occurs with probability at least $1/4$.
Since $\gamma \cdot \gamma'$ occurs within $O(n(d_1+d_2))$ steps
in expectation by Lemma \ref{lem:occurs},
we can conclude that $\ex[\tdel'] = O(n(d_1+s_2))$.
\end{proof}

\begin{lemma}
\label{lem:srl} 
$\max_{C \in \call(\cni)} \ect(C,\srl)=O(n^2)$.
\end{lemma}

\begin{proof}
Let $C_0$ be any configuration in $\cni$
and let $\Xi = \Xi_{\prl}(C_0,\sch)$.
By Lemmas \ref{lem:cni_closed} and \ref{lem:noleader},
the number of leaders is monotonically non-increasing
and never becomes zero in $\Xi$.
For any real number $x$,
define $\calL_{x}$ as the set of configurations 
where the number of leaders is at most $x$.
First, we prove the following claim.
\begin{clm}
\label{clm:decrease}
Let $\alpha = 12/11$.
For any sufficiently large integer $k=O(1)$,
if $C_0 \in (\calL_{\alpha^{k+1}} \setminus \calL_{\alpha^{k}})
\cap \cni$,
execution $\Xi$ reaches a configuration in $\calL_{\alpha^{k}} \cap \cni$
within $O(n^2/\alpha^{k})$ steps in expectation.
\end{clm}
\begin{proof}
Let $l_0,l_1,\dots,l_{s-1}=u_{\pi_0},u_{\pi_1},\dots,u_{\pi_{s-1}}$ be the leaders in $C$,
where $\pi_0 < \pi_1 < \dots < \pi_{s-1}$.
We say that $l_j$ and $l_{j+1 \bmod s}$ are neighboring leaders
for each $j=0,1,\dots,s-1$.
Since there are $s$ leaders in $C_0$,
there are at least $3s/4$ leaders $l_i=u_{\pi_i}$ such that
$\distl(\pi_{i+1 \bmod s}) \le 4n/s$.
Thus,
there are at least $n/2$ leaders $l_j=u_{\pi_j}$
such that $\distl(\pi_{j+1 \bmod s}) \le 4n/s$
and $\distl(\pi_{j+2} \bmod s) \le 4n/s$ in $C_0$.
Let $S_L$ be the set of all such leaders.
For each $l_j \in S_L$,
by Lemma \ref{lem:killed} and Markov inequality,
$l_j$, $l_{j+1 \bmod s}$, or $l_{j+2 \bmod s}$
becomes a follower within $O(n^2/s)$ steps
with probability $1/2$.
Generally, if $X_0, X_1, \dots, X_i$ are 
(possibly non-independent) events each of which
occurs with probability at least $1/2$,
no less than half of the events occurs 
with probability at least $1/2$.
Thus, with probability $1/2$, 
for no less than half of the leaders $l_j$ in $S_L$,
the event that $l_j$, $l_{j+1 \bmod s}$, or $l_{j+2 \bmod s}$
becomes a follower occurs within $O(n^2/s)$ steps.
Thus, at least $|S_L| \cdot (1/2) \cdot (1/3)=s/12$ leaders
becomes followers within $O(n^2/s)$ steps
with probability at least $1/2=\Omega(1)$. 
Repeating this analysis,
we observe that $\Xi$ reaches
a configuration in $\calL_{\alpha^{k}} \cap \cni$
within $O(n^2/\alpha^{k})$ steps in expectation.
\end{proof}

By Claim \ref{clm:decrease},
for sufficiently large integer $k=O(1)$,
the number of leaders becomes a constant
(\ie $O(\alpha^k)=O(1)$) within
$\sum_{i=k}^{\lceil \log_\alpha n \rceil}O(n^2/\alpha^i)=O(n^2)$
steps in expectation in $\Xi$.
Thereafter, by Lemma \ref{lem:killed}, 
the number of leaders decreases to one within $O(n^2)$ steps
in expectation.
\end{proof}

Lemmas \ref{lem:srl_closed}, \ref{lem:cni}, and \ref{lem:srl}
give the following main theorem. 
\begin{theorem}
Given an integer $N$, $\prl$ is a self-stabilizing 
leader election protocol for any directed rings of any size $n \le N$. 
The convergence time is $O(nN)$.
The number of states is $O(N)$.
\end{theorem}

\section{Conclusion}
We have presented a self-stabilizing leader election protocol for rings in the population protocol.
Given the knowledge of an upper bound $N$ of the population size $n$, the proposed protocol solves a self-stabilizing leader election problem for any directed ring.
Specifically, an execution of the protocol starting from any initial configuration elects a unique leader within $O(nN)$ steps
in expectation, by using $O(N)$ states per agent. If a given knowledge $N$ is asymptotically tight, \ie $N=O(n)$, this protocol is time-optimal.

%
%
%
%

\bibliographystyle{plain} 
\bibliography{biblio} 

\begin{thebibliography}{10}

\bibitem{AAF+08}
D.~Angluin, J.~Aspnes, M.~J Fischer, and H.~Jiang.
\newblock Self-stabilizing population protocols.
\newblock {\em ACM Transactions on Autonomous and Adaptive Systems},
  3(4):1--28, 2008.

\bibitem{AAD+06}
Dana Angluin, James Aspnes, Zo\"{e} Diamadi, Michael~J. Fischer, and Ren\'{e}
  Peralta.
\newblock Computation in networks of passively mobile finite-state sensors.
\newblock {\em Distributed Computing}, 18(4):235--253, 2006.

\bibitem{AAE08}
Dana Angluin, James Aspnes, and David Eisenstat.
\newblock Fast computation by population protocols with a leader.
\newblock {\em Distributed Computing}, 21(3):183--199, 2008.

\bibitem{BBB13}
J.~Beauquier, P.~Blanchard, and J.~Burman.
\newblock Self-stabilizing leader election in population protocols over
  arbitrary communication graphs.
\newblock In {\em International Conference on Principles of Distributed
  Systems}, pages 38--52, 2013.

\bibitem{Bur+19}
Janna Burman, David Doty, Thomas Nowak, Eric~E Severson, and Chuan Xu.
\newblock Efficient self-stabilizing leader election in population protocols.
\newblock {\em arXiv preprint arXiv:1907.06068}, 2019.

\bibitem{SIW12}
S.~Cai, T.~Izumi, and K.~Wada.
\newblock How to prove impossibility under global fairness: On space complexity
  of self-stabilizing leader election on a population protocol model.
\newblock {\em Theory of Computing Systems}, 50(3):433--445, 2012.

\bibitem{CP07}
D.~Canepa and M.~G. Potop-Butucaru.
\newblock Stabilizing leader election in population protocols.
\newblock {\em http://hal.inria.fr/inria-00166632}, 2007.

\bibitem{CC19}
Hsueh-Ping Chen and Ho-Lin Chen.
\newblock Self-stabilizing leader election.
\newblock In {\em Proceedings of the 38th ACM Symposium on Principles of
  Distributed Computing}, pages 53--59, 2019.

\bibitem{CC20}
Hsueh-Ping Chen and Ho-Lin Chen.
\newblock Self-stabilizing leader election in regular graphs.
\newblock In {\em Proceedings of the 39th Symposium on Principles of
  Distributed Computing}, pages 210--217, 2020.

\bibitem{CG17}
Gennaro Cordasco and Luisa Gargano.
\newblock Space-optimal proportion consensus with population protocols.
\newblock In {\em International Symposium on Stabilization, Safety, and
  Security of Distributed Systems}, pages 384--398, 2017.

\bibitem{Dij74}
E.W. Dijkstra.
\newblock {Self-stabilizing systems in spite of distributed control}.
\newblock {\em Communications of the ACM}, 17(11):643--644, 1974.

\bibitem{FJ06}
M.~J. Fischer and H.~Jiang.
\newblock Self-stabilizing leader election in networks of finite-state
  anonymous agents.
\newblock In {\em International Conference on Principles of Distributed
  Systems}, pages 395--409, 2006.

\bibitem{Izu15}
T.~Izumi.
\newblock On space and time complexity of loosely-stabilizing leader election.
\newblock In {\em International Colloquium on Structural Information and
  Communication Complexity}, pages 299--312, 2015.

\bibitem{MNRS14}
George~B Mertzios, Sotiris~E Nikoletseas, Christoforos~L Raptopoulos, and
  Paul~G Spirakis.
\newblock Determining majority in networks with local interactions and very
  small local memory.
\newblock In {\em International Colloquium on Automata, Languages, and
  Programming}, pages 871--882, 2014.

\bibitem{SOK+14}
Y.~Sudo, F.~Ooshita, H.~Kakugawa, and T.~Masuzawa.
\newblock Loosely-stabilizing leader election on arbitrary graphs in population
  protocols.
\newblock In {\em International Conference on Principles of Distributed
  Systems}, pages 339--354, 2014.

\bibitem{SEI+20}
Yuichi Sudo, Ryota Eguchi, Taisuke Izumi, and Toshimitsu Masuzawa.
\newblock Time-optimal loosely-stabilizing leader election in population
  protocols.
\newblock {\em arXiv preprint arXiv:2005.09944}, 2020.

\bibitem{SMD+16}
Yuichi Sudo, Toshimitsu Masuzawa, Ajoy~K Datta, and Lawrence~L Larmore.
\newblock The same speed timer in population protocols.
\newblock In {\em the 36th IEEE International Conference on Distributed
  Computing Systems}, pages 252--261, 2016.

\bibitem{SNY+12}
Yuichi Sudo, Junya Nakamura, Yukiko Yamauchi, Fukuhito Ooshita, Hirotsugu.
  Kakugawa, and Toshimitsu Masuzawa.
\newblock Loosely-stabilizing leader election in a population protocol model.
\newblock {\em Theoretical Computer Science}, 444:100--112, 2012.

\bibitem{SOK+20det}
Yuichi Sudo, Fukuhito Ooshita, Hirotsugu Kakugawa, and Toshimitsu Masuzawa.
\newblock Loosely stabilizing leader election on arbitrary graphs in population
  protocols without identifiers or random numbers.
\newblock {\em IEICE Transactions on Information and Systems}, 103(3):489--499,
  2020.

\bibitem{SOK+18}
Yuichi Sudo, Fukuhito Ooshita, Hirotsugu Kakugawa, Toshimitsu Masuzawa, Ajoy~K
  Datta, and Lawrence~L Larmore.
\newblock Loosely-stabilizing leader election for arbitrary graphs in
  population protocol model.
\newblock {\em IEEE Transactions on Parallel and Distributed Systems},
  30(6):1359--1373, 2018.

\bibitem{SOK+20polylog}
Yuichi Sudo, Fukuhito Ooshita, Hirotsugu Kakugawa, Toshimitsu Masuzawa, Ajoy~K
  Datta, and Lawrence~L Larmore.
\newblock Loosely-stabilizing leader election with polylogarithmic convergence
  time.
\newblock {\em Theoretical Computer Science}, 806:617--631, 2020.

\bibitem{SSN+20}
Yuichi Sudo, Masahiro Shibata, Junya Nakamura, Yonghwan Kim, and Toshimitsu
  Masuzawa.
\newblock The power of global knowledge on self-stabilizing population
  protocols.
\newblock In {\em International Colloquium on Structural Information and
  Communication Complexity}, pages 237--254, 2020.

\end{thebibliography}


\end{document}